%% file: MEF-arXiv.tex
\documentclass[11pt]{article}

\usepackage{algorithmic,algorithm}
\usepackage{color,multicol,graphics,epsfig,epstopdf}
\usepackage{amsmath}
\usepackage{amsthm}
\usepackage{amssymb}
\usepackage{float}
\usepackage{graphicx,url}
\usepackage{cite}
\usepackage{subfigure}
\usepackage{tikz}
\usetikzlibrary{patterns}
\usepackage{times}
\usepackage{multirow}

\newenvironment{AvoidOverfullParagraph}[0]
{\sloppy\ignorespaces}
{\par\fussy\ignorespacesafterend}

\usetikzlibrary{arrows,decorations,decorations.shapes,backgrounds,shapes,decorations.pathreplacing}

\long\def\LongVersion#1\LongVersionEnd{#1}
\long\def\ShortVersion#1\ShortVersionEnd{}
\long\def\ExtraVersion#1\ExtraVersionEnd{}

\newtheorem{theorem}{Theorem}[section]
\newtheorem{corollary}[theorem]{Corollary}
\newtheorem{lemma}[theorem]{Lemma}
\newtheorem{claim}[theorem]{Claim}

\newtheorem{example}{Example}

\floatstyle{ruled}
\newfloat{algorithm}{tbp}{loa}
\floatname{algorithm}{Algorithm}

\newcommand{\acron}[0]{MC-CWE}
\newcommand{\items}{M}
\newcommand{\auction}{A}
\newcommand{\agents}{N}

\newcommand{\demand}{D}

\newcommand{\demandi}[1][i]{{\demand_{#1}}}

\newif\iftechnicalreport
\technicalreportfalse
\technicalreporttrue

\newcommand{\comment}[1]{}

\newcommand{\be}{\begin{equation}}
\newcommand{\ee}{\end{equation}}

\newcommand{\argmin}{\mathop{\rm argmin}}

\newcommand{\argmax}{\mathop{\rm argmax}}

\newcommand{\mb}[1]{\mathbf{#1}}

\def \reals {{\mathbb R}}

\newcommand{\val}{v}
\newcommand{\vals}{{\mathbf \val}}

\newcommand{\vali}[1][i]{{\val_{#1}}}

\newcommand{\util}{u}

\newcommand{\utili}[1][i]{{\util_{#1}}}

\newcommand{\price}{p}
\newcommand{\prices}{{\mathbf \price}}
\newcommand{\pricei}[1][i]{{\price_{#1}}}

\newcommand{\alloc}{x}
\newcommand{\allocs}{{\mathbf \alloc}}

\newcommand{\alloci}[1][i]{{\alloc_{#1}}}

\newcommand{\partition}{{\mathbf A}}

\begin{document}

\title{Clearing Markets via Bundles}

\author{Michal Feldman\thanks{mfeldman@tau.ac.il. Tel-Aviv University.} \and Brendan Lucier\thanks{blucier@microsoft.com. Microsoft Research New England.}}

\maketitle


\maketitle

\begin{abstract}
\input{abstract}
\end{abstract}


\input{MEF-intro.tex}

\input{MEF-prelim.tex}
\input{super-additive.tex}

\input{budget-additive.tex}

\bibliography{MEF-Bib}
\bibliographystyle{abbrv}
\appendix

\input{appendix}

\end{document}

%% file: abstract.tex
We study algorithms for combinatorial market design problems, where a set of heterogeneous and indivisible objects are priced and sold to potential buyers subject to equilibrium constraints.
Extending the CWE notion introduced by Feldman et al. [STOC 2013],
we introduce the concept of a \emph{Market-Clearing Combinatorial Walrasian Equilibium} (MC-CWE) as a natural relaxation of the classical Walrasian equilibrium (WE) solution concept.
The only difference between a MC-CWE and a WE is the ability for the seller to bundle the items prior to sale.
This innocuous and natural bundling operation imposes a plethora of algorithmic and economic challenges and opportunities.
Unlike WE, which is guaranteed to exist only for (gross) substitutes valuations, a MC-CWE always exists.
The main algorithmic challenge, therefore, is to design computationally efficient mechanisms that generate MC-CWE outcomes that approximately maximize social welfare.
For a variety of valuation classes encompassing substitutes and complements (including super-additive, single-minded and budget-additive valuations), we design polynomial-time MC-CWE mechanisms that provide tight welfare approximation results.

%% file: MEF-intro.tex
\section{Introduction}

\begin{AvoidOverfullParagraph}
The {\em resource allocation} problem lies at the heart of theoretical economics:
how should scarce resources be allocated among individual agents with competing interests?
Since the emergence of the Internet, which enables complex resource allocation on a grand scale, this has naturally become a central problem in computer science as well.
Economists generally approach this problem by adopting the notion of {\em market equilibrium}.
Broadly speaking, a market equilibrium is a set of resource prices that are stable in the sense that all agents are maximally happy with their allocations and no resources are left unallocated.
A long line of work has been dedicated to addressing the existence of equilibrium prices,
and it has been shown (see, e.g., \cite{ArrowDebreu54}) that market equilibria exist very generally, as long as the market is {\em convex}.
\end{AvoidOverfullParagraph}

While this result sounds appealing in its generality,
the convexity assumption usually requires that resources be infinitely divisible.
In many applications of interest, especially those with a computational aspect, resources are indivisible; in these cases the convexity assumption is inapplicable.
Do the results from the convex environments carry over to non-convex environments?
In general the answer is no: the existence of equilibrium prices is not guaranteed.  As a result, the study of markets for indivisible goods tends to focus on specific cases for which such prices exists, such as when buyer values satisfy the \emph{gross substitutes} condition.

To be more precise, consider the following concrete model.
Suppose there are $m$ indivisible and heterogeneous items that should be allocated among $n$ agents.
The agents have potentially arbitrary preferences over bundles of items.
Formally, every agent $i \in [n]$ has a valuation function $\vali$ that maps every subset $S$ of items
into the value $\vali(S)$ that agent $i$ derives from the bundle $S$ (in monetary terms).
Given a price vector $\prices=(\price_1, \ldots, \price_m)$, a bundle $S$ is said to be in agent $i$'s {\em demand set} if $S$ maximizes $i$'s {\em utility} given $\prices$, defined as the difference between $\vali(S)$ and $\sum_{j \in S}p_j$.

A {\em Walrasian equilibrium} (WE) is an assignment of {\em item prices} to the $m$ items, and an assignment of the objects to the agents, such that: (1) every agent is allocated a bundle in his demand set, and (2) the market {\em clears}; i.e., all items are allocated\footnote{More precisely, every unallocated item is priced at zero.}.
Such a solution is truly appealing; every agent is maximally happy despite competing preferences,
the market clears, and the pricing structure is natural, simple, and transparent.
Unfortunately, WE do not exist in general.  A WE is guaranteed to exist only for a rather narrow class of valuations, known as gross substitutes (GS) valuations
(a strict subset of submodular functions) \cite{Gul1999}.
This eliminates any hope for the applicability of WE to environments with valuations that exhibit complementarities, and many forms of substitutes as well.

Recently, Feldman et al. 
\cite{Feldman2013}
proposed a 
relaxation of WE, termed a {\em combinatorial Walrasian equilibrium} (CWE).  In a CWE, the seller can choose to {\em bundle} objects prior to assigning prices.
This is a natural power to afford the seller, since as the owner of the resources he has some inherent ability to define what is meant by an ``item.''
The generated bundles induce a reduced market --- a market in which the items for sale are the bundles generated by the seller.
In addition to the bundling operation, the CWE further relaxes the WE notion in that it allows for items to remain unallocated (even when they are priced above zero).
It is easy to see that a CWE exists for any valuation profile, since the seller could bundle all objects into a single item.
The important issue, then, is whether there exists a CWE that is (approximately) efficient with respect to social welfare.
Indeed,
Feldman et al. 
show there always exists a CWE with at least half of the optimal (unconstrained) social welfare \cite{Feldman2013}.

The CWE notion relaxes the WE notion in two ways: (i) it allows bundling, (ii) it does not require market clearance.
While the bundling relaxation is central to the notion of CWE, the second relaxation warrants some discussion.
The relaxation of market clearance is somewhat at odds with the notion of a two-sided market equilibrium:
prices might not be stable from the seller's perspective.
After all, if an object (i.e.\ bundle) does not sell, the seller may be tempted to decrease its price in order to to sell it and increase revenue.  The concept of CWE therefore implicitly requires that the seller pre-commit to (sub-optimal) prices, in addition to committing to a bundling of the items.
With this in mind, we consider whether the relaxation of market clearance is truly necessary.  It is easy to see that the bundling relaxation alone is enough to guarantee existence of an equilibrium, so the question becomes one of welfare.
Can we hope to achieve the welfare bound of
\cite{Feldman2013}
without relaxing market clearance?




To answer this question we define the notion of a {\em Market-clearing CWE} (MC-CWE),
which allows the bundling operation, but requires market clearance.
A \acron{} is, therefore, precisely a WE over the reduced market; it differs from a WE only in the ability of the seller to pre-bundle the items,
and in particular it is a stronger (more restrictive) concept than CWE.

For a number of valuation classes, encompassing both substitutes and complements, we provide two types of results.
The first finds the fraction of the optimal social welfare that can be obtained in a \acron{} outcome.
The second addresses the same problem but under the additional requirement of operating in polynomial time.
Note that the approximation result established in
\cite{Feldman2013}
is only semi-computational --- given the optimal allocation, it finds in polynomial time a CWE outcome that gives at least a half of the optimal welfare.
Here, we devise polynomial approximation algorithms that do not need any initial allocation.
Moreover, all of our approximation results match the computational lower bounds for their corresponding valuation classes.


We note that while the focus of our paper is welfare maximization,
our analysis and results have some immediate implications on revenue maximization.
In particular, all of our approximation results carry over to revenue approximation,
since for each of our mechanisms the bound we obtain on the social welfare is precisely the revenue extracted at equilibrium.
We also observe that from a revenue maximization perspective, \acron{} might be preferred over WE even for GS valuations,
as it can extract \emph{arbitrarily} higher revenue than any WE solution (see Appendix \ref{app:revenue}).

\subsection{Our Results and Techniques}



\setlength{\tabcolsep}{2pt}
\begin{figure}
\begin{center}
\small{
\begin{tabular}{|c|c|c|c|c|}
\hline
\multirow{2}{*}{}&Uniform BA&Uniform BA&Single minded&Super additive\\
& identical budgets & & &\\
\hline
\multirow{2}{*}{ \acron{} gap}&1&$\geq 8/7$&1&1\\
& & $\leq2$  & &\\
\hline
\multirow{2}{*}{Poly-time \acron{} approx.} & $\leq 4/3$ & $ \leq 8/3$ & $O(m/\sqrt{\log m})$ [value] & $\theta(\sqrt{m})$\\
& & & $\theta(\sqrt{m})$ [demand] &\\
\hline
\end{tabular}
}
\label{tab:results}
\caption{Summary of our approximation results. The columns correspond to valuation classes. The first column corresponds to uniform budget additive valuations with identical budgets, and the second column corresponds to uniform budget additive valuations with arbitrary budgets.
The first row corresponds to the gap in social welfare due to \acron{}, disregarding computational considerations.
The second row corresponds to the approximation that can be achieved with a \acron{} poly-time mechanism.
All approximation results assume the value-query model, unless otherwise stated.
Note that $m$ is the number of items for sale.}
\end{center}
\end{figure}


We construct \acron{} mechanisms with certain welfare guarantees for various valuation classes.
These results are stated below and summarized in Table~\ref{tab:results}.

\paragraph{Super-additive valuations}
In the case where agent valuations are super-additive, we show that
there always exists a \acron{} that maximizes social welfare.
Note that it is not always possible to maximize social welfare without bundling, even if the market clearance requirement is relaxed: there exist input instances in which all bidders are single-minded, but every outcome with item pricing obtains only an $O(\sqrt{m})$ fraction of the optimal social welfare \cite{FGL-13}.  The use of bundling is therefore crucial in generating a socially efficient equilibrium outcome.

%
We next turn to computational algorithms.
We show how to construct a \acron{} that obtains an $O(\sqrt{m})$ approximation to the optimal social welfare in a polynomial number of \emph{demand queries}\footnote{A demand query returns the utility-maximizing set given a vector of item prices.}, matching known lower bounds \cite{Nisan2006}.  This mechanism is new; as far as we are aware, existing $O(\sqrt{m})$-approximations do not satisfy the conditions of \acron{}.  Our mechanism proceeds by first crafting an $O(\sqrt{m})$-approximate allocation and prices, then applying local search to repeatedly satisfy agent demands (bundling objects and/or raising prices in the process) until every agent obtains a demanded set at the given prices.  Our construction makes use of demand queries in a way similar to that of \cite{Feldman2013}: rather than querying demand sets over the original space of objects, we query demand over bundles of objects (under linear prices).
With {\em value queries}, we show that the $O(m/\sqrt{\log m})$-approximate mechanism due to \cite{HKNS2004} satisfies the \acron{} property.
We also show that in the case of {\em single-minded valuations}, our demand-query mechanism can be modified to achieve an $O(\sqrt{m})$ approximation in a polynomial number of value queries.

\paragraph{Sub-additive valuations: uniform budget additive}
We then turn to the space of sub-additive valuations.
Since efficient WE exist for the class of GS valuations, efficient \acron{} exist for this class as well.  We therefore consider a class of non-GS valuations: those that are {\em uniform budget-additive}.  In this auction problem, each item $a$ has a common value $v_a$, and each agent values the item at either $0$ or $v_a$.  Furthermore, each bidder has a budget that limits his value for any set of items.
For this class, we demonstrate that WE do not necessarily exist.
Moreover, we provide an instance in which no \acron{} can achieve more than a $7/8$ fraction of the optimal social welfare.
On the other hand, we show any allocation can be converted (in polynomial time) into a \acron{} outcome that achieves at least half of the original social welfare.
Thus, at least half of the optimal welfare can always be achieved in a \acron{} outcome.

Turning to computational consideration, the welfare-maximization problem for this valuation class is known to be APX-hard,
and the best-known algorithm achieves an approximation of $4/3$ (see \cite{Andelman04,Azar08,Srinivasan2008,Chakrabarty2008}).
Combined with the aforementioned algorithm, which converts every outcome to a \acron{} outcome with at least half of the welfare of the original outcome, this implies a \acron{} mechanism that achieves a $8/3$ approximation.
Our analysis is based on the observation that an outcome can be implemented at \acron{} if and only if that outcome is an optimal solution (among all fractional solutions) to a certain linear program: the {\em configuration LP} for the assignment problem restricted to the bundles in the outcome allocation.
Our construction is based upon local search, but of a different nature than our super-additive mechanism.  Rather than attempting to improve social welfare, we repeatedly bundle objects to \emph{reduce} the optimal \emph{fractional} welfare, shrinking the gap between fractional and integral solutions to the configuration LP.


We further show that if agents have identical budgets, the factor-$2$ loss disappears: any allocation can be made \acron{} without loss in social welfare. Yet, even within this restricted class, a Walrasian equilibrium may not exist.
These results are driven by connections between \acron{} and the configuration LP for the combinatorial assignment problem.

\subsection{Relation to Prior Work}
There is a long line of work studying pricing equilibria in theoretical economics.
\cite{Foley1967} and
\cite{Varian1974} initiated the study of envy-freeness outcomes.  Market-clearing prices in the market assignment problem were studied by 
\cite{Shapley1971}.  Subsequent works have studied conditions for the existence of Walrasian equilibria \cite{Aumann1975,Kelso1982,Leonard1983,Bikhchandani1997,Gul1999}.




An alternative to Walrasian equilibrium is to allow a seller to set (non-linear) prices on arbitrary bundles.
Such package auctions were formalized by
\cite{Bikhchandani2002}.
Some notable examples of combinatorial auction mechanisms that make use of bundle prices can be found in
\cite{Ausubel2002},
\cite{Wurman2000}, and
\cite{Parkes2000}.
Our notion of \acron{} differs from package auctions in that the seller commits to a partition of the objects, then sets linear prices over those bundles.




\cite{Fiat2009} study an extension of envy-freeness, multi-envy-freeness, in which no agent envies any subset of other agents.
This concept is related to our notion of \acron{}.
However, crucially, they restrict their definition to agents with single-minded types, which dampens the distinction from envy-freeness. 




Many of our mechanisms access agent valuations via \emph{demand queries}, whereby agents provide feedback to sellers through their choice of demand set at a given price vector.
See \cite{Blumrosen2005,Mirrokni2008,Dobzinski2011,Oren2012} for discussions on the power of the demand query model. 

The assignment problem with budget-additive bidders has received significant attention in purely algorithmic frameworks.
\cite{Andelman04} present an algorithm with approximation factor $(1-1/e)$, and also introduce the value-uniform variant of the problem.  This approximation factor was subsequently improved \cite{Azar08,Srinivasan2008},
leading eventually to a $4/3$ approximation algorithm due to
\cite{Chakrabarty2008}.

In a related paper,
\cite{FKL-12} introduce the notion of \emph{conditional equilibrium}, a relaxation of WE in which no buyer wishes to add additional items to his allocation under the given prices.  They show that every conditional equilibrium achieves at least half of the optimal social welfare, and a conditional equilibrium always exists when buyers have submodular valuations.  Our equilibrium concept \acron{} differs in that it is based on reducing the space of objects via bundling, rather than relaxing the notion of envy-freeness directly.

%% file: MEF-prelim.tex
\section{Preliminaries}
\label{sec.prelim}

The auction setting considered in this work consists of a set $\items$ of $m$ indivisible objects and a set of $n$ agents.
Each agent has a valuation function $\vali(\cdot) : 2^\items \to \reals_{\geq 0}$ that indicates his value for
every set of objects, is non-decreasing (i.e., $\vali(S) \leq \vali(T)$ for every $S \subseteq T \subseteq \items$) and is normalized so that $\vali(\emptyset) = 0$.
The profile of agent valuations is denoted by $\vals=(\val_1,\dotsc,\val_n)$, and an auction setting is defined by a tuple $\auction=(\items,\vals)$.

A price vector $\prices=(\price_1, \dotsc, \price_m)$ consists of a price $\price_j$ for each object $j \in M$. An {\em allocation} is a vector of sets $\allocs = (\alloc_0, \alloc_1, \dotsc, \alloc_n)$, where $\alloc_i \cap \alloc_k = \emptyset$ for every $i\neq k$, and $\bigcup_{i=0}^{n}\alloci = \items$. In the allocation $\allocs$, for every $i\in\agents$, $\alloci$ is the bundle assigned to agent $i$, and $\alloc_0$ is the set of unallocated objects; i.e., $\alloc_0 = \items \setminus \bigcup_{i=1}^{n}\alloci$.

As standard, we assume that each agent has a quasi-linear utility function; i.e., the utility of agent $i$ being allocated bundle $\alloc_i$ under prices $\prices$ is $\utili(\alloci, \prices) = \vali(\alloci) - \sum_{j \in \alloci}\price_j.$
Given prices $\prices$, the {\em demand correspondence} $\demandi(\prices)$ of agent $i$ contains the sets of objects that maximize agent $i$'s utility:
\[
\demandi(\prices) = \left\{S^*:  S^* \in\argmax_{S\subseteq M}\{\utili(S,\prices)\}\right\}.
\]
A tuple $(\allocs,\prices)$ is said to be {\em buyer stable} for auction $\auction=(M,\vals)$
if $\alloci \in\demandi(\prices)$ for every $i\in\agents$.
A tuple $(\allocs,\prices)$ is said to be {\em seller stable} for auction $\auction=(M,\vals)$
if for every $j \in \alloc_0$, $\price_j=0$.
The seller stability condition is also known as {\em market clearance}.

A tuple $(\allocs,\prices)$ is said to be a {\em Walrasian equilibrium} (WE) for auction $\auction=(M,\vals)$
if it is both buyer stable and seller stable.

It is well known that a WE exists only under restricted valuation functions.
In particular, the class of \emph{gross substitutes} (GS) valuations is a maximal class that admits Walrasian equilibria \cite{Gul1999}.
This class is a strict subset of submodular valuations.
It is also known that if a WE exists, it is economically efficient (i.e., maximizes social welfare --- the sum of agents' valuations).


We next define the notion of a {\em Market-Clearing CWE} (\acron{}).
The crux of the concept is that items are pre-partitioned into indivisible bundles.
The constructed bundles are treated as indivisible objects, and the \acron{} notion reduces to WE over the bundles.
Crucially, although prices are now associated with bundles (of the original market), unlike previous notions,
prices are linear once bundles are fixed.
A formal definition follows.

For a partition $\partition=(\partition_1, \ldots, \partition_k)$ of the item set $M$ we slightly
abuse notation and denote by $\partition=\{\partition_1, \ldots, \partition_k\}$ the reduced set of items, where the
valuation of each agent $i$ of a subset $S \subseteq \partition$ is $\vali(\bigcup_{j:\partition_j \in S}\partition_j)$. We
denote by $\auction_{\partition}$ an auction over the reduced set of items $\partition$ with the induced valuation profile.

Every allocation $\allocs$ induces a partition $\partition(\allocs) = (\alloc_0, \ldots, \alloc_n)$.
A tuple $(\allocs,\prices)$, where $\allocs=(\alloc_0, \ldots, \alloc_n)$, and $\pricei$ is the price of $\alloci$ for every $\alloci\neq\emptyset$, is a {\em Market-Clearing CWE} (\acron{}) if $(\allocs,\prices)$ is a WE in the auction $\auction_{\partition(\allocs)}$.
An allocation $\allocs$ is said to be \acron{} if it admits a price vector $\prices \in \reals_{\ge 0}^{n+1}$ such that $(\allocs,\prices)$ is \acron{}. A mechanism is said to be \acron{} if it maps every valuation profile $\vals$ to an outcome $(\allocs,\prices)$ that is \acron{}.

\noindent {\bf Relation to Combinatorial Walrasian equilibrium (CWE)}

A tuple $(\allocs,\prices)$, where $\allocs=(\alloc_0, \ldots, \alloc_n)$, and $\pricei$ is the price of $\alloci$ for every $\alloci\neq\emptyset$, is a CWE if $(\allocs,\prices)$ is buyer stable in the auction $\auction_{\partition(\allocs)}$.
Note that a \acron{} is weaker than a WE, since it allows for a pre-sale partition of the goods.
On the other hand, \acron{} is stronger than a CWE, since it requires seller stability on top of buyer stability.

The \acron{} concept is also closely related to the notion of envy-freeness.
See more on this relation in Appendix~\ref{sec:ef}.

\subsection{Characterization}

The characterization of a CWE allocation is closely related to the characterization of an allocation that can be supported in a WE~\cite{Bikhchandani1997}.
A similar observation was already given in \cite{Feldman2013}, but we state it here for completeness, as it is used in later sections.

For a given partition $\partition$ of the objects, the allocation of $\partition$ to $\agents$ can be specified by a set of
integral variables $y_{_{i,S}}\in \{0,1\}$, where $y_{_{i,S}}=1$ if the set $S\subseteq\partition$ is allocated to agent $i\in\agents$ and
$y_{_{i,S}}=0$ otherwise.
These variables should satisfy the following conditions: $\sum_S y_{_{i,S}}\leq 1$ for every $i \in \agents$
(each agent is allocated to at most one bundle) and $\sum\limits_{i, S \supseteq \partition_j} y_{_{i,S}} \leq 1$ for every $\partition_j \in \partition$ (each element of the partition is allocated to at most one agent).
A \emph{fractional allocation} of $\partition$ is given by variables $y_{_{i,S}} \in [0,1]$ that satisfy the same conditions and
intuitively might be viewed as an allocation of divisible items.
The configuration LP for $\auction_{\partition}$ is given by the following linear program, which computes
the fractional allocation that maximizes social welfare.
\begin{eqnarray}
\max & & \sum_{i,S} \vali(S)\cdot y_{_{i,S}} \nonumber\\
\mbox{s.t.} & & \sum_S y_{_{i,S}} \leq 1 \mbox{ for every } i \in \agents \nonumber\\
& & \sum_{i, S \supseteq \partition_j} y_{_{i,S}} \leq 1 \mbox{ for every } \partition_j \in \partition \nonumber\\
& & y_{_{i,S}} \in [0,1] \mbox{ for every } i \in \agents, S \subseteq \partition \nonumber
\end{eqnarray}
The characterization given in \cite{Bikhchandani1997} states that a WE exists if and only if the optimal fractional
solution to the allocation LP occurs at an integral solution.
This characterization of a WE allocation can be used to derive a characterization of a \acron{} allocation.

Recall that every allocation $\allocs$ induces a partition $\partition(\allocs) = (\alloc_0, \ldots, \alloc_n)$.
The WE characterization implies the following \acron{} characterization.
\begin{claim}
\label{claim.frac}
An allocation $\allocs=(\alloc_0, \alloc_1, \ldots, \alloc_n)$ is a \acron{} for $\auction$ iff
the configuration LP for $\auction_{\partition(\allocs)}$ has an integral optimal solution that sets
$y_{_{i,\alloci}} = 1$ for all $i \in \agents$.
\end{claim}

In light of Claim \ref{claim.frac}, the problem of finding a \acron{} allocation is equivalent to the problem of
finding bundles such that the optimal social welfare generated by fractional and integral allocations
\emph{of these bundles} are identical, then returning an efficient allocation of those bundles.

\paragraph{Gap in social welfare due to \acron{}: a lower bound.}
In \cite{Feldman2013}, a lower bound on the social welfare of any CWE allocation is established.
Specifically, using the equivalence of Claim \ref{claim.frac}, it is shown that
the welfare of any CWE allocation is at most $2/3$ of the welfare of an
optimal (non-CWE) allocation.
We note that the same example trivially carries over to the stronger notion of \acron{};
i.e., the welfare of any \acron{} allocation is at most $2/3$ of the optimal welfare.
The established lower bound holds for general valuations, but disappears for some families of valuation functions (as will be shown in later sections of this paper).

%% file: super-additive.tex
\section{Super-additive valuations}
\label{sec:super-add}

In this section we study \acron{} outcomes when agent valuations are super-additive.
We first show that there is no loss in efficiency due to \acron{}.  In particular, every efficient allocation is \acron{}.

We say that allocation $\allocs$ is \emph{bundle-efficient} for $\mb{v}$ if, for all functions $\beta : [n] \to [n]$, we have
$\sum_i v_i(\alloc_i) \geq \sum_i v_i\left( \bigcup_{j \in \beta^{-1}(i)} \alloc_j \right).$
That is, a bundle-efficient allocation $\allocs$ maximizes social welfare among all ways to allocate the (indivisible) bundles $\alloc_1, \dotsc, \alloc_n$ to the agents.

\begin{theorem}
\label{thm.superadd}
If agents are super-additive and $\allocs$ is a bundle-efficient allocation, then the price vector $\pricei = \vali(\alloci)$ is such that $(\allocs, \prices)$ is \acron{}.
\end{theorem}

\begin{proof}
Pick any agent $i$ and set of agents $S$.
If $i \in S$, then
\[
 \vali(\alloci) - \pricei = 0 \geq \vali(\cup_{j \in S} \alloc_j) - \sum_{j \in S} \val_j(\alloc_j) = \vali(\cup_{j \in S} \alloc_j) - \sum_{j \in S} \price_j,
\]
where the inequality follows from the efficiency of allocation $\allocs$.
Consider next the case where $i \not\in S$.
It holds that
\[
 \vali(\alloci) + \sum_{j \in S} \val_j(\alloc_j) \geq \vali(\cup_{j \in S \cup \{i\}} \alloc_j) \geq \vali(\alloci) + \vali(\cup_{j \in S} \alloc_j),
\]
where the first inequality follows from the efficiency of allocation $\allocs$, and the second inequality follows from the super additivity of $\vali$. It follows that
\[
\vali(\alloci) - \pricei = 0 \geq \vali(\cup_{j \in S} \alloc_j) - \sum_{j \in S} \val_j(\alloc_j) = \vali(\cup_{j \in S} \alloc_j) - \sum_{j \in S} \price_j.
\]
The assertion of the theorem follows.
\end{proof}

A consequence of Theorem \ref{thm.superadd} is that the full surplus can always be extracted by the seller as revenue. 


We note that the use of bundling is necessary for the statement of Theorem \ref{thm.superadd}, even if we relax the market-clearing requirement of Walrasian equilibrium and even for single-minded bidders.  There are known auction instances for which, given any set of item prices for which agent demand sets are disjoint, the social welfare of the resulting allocation is an $O(\sqrt{m})$ fraction of the optimal social welfare \cite{FGL-13}.  For completeness we describe such an example in Appendix \ref{app.bundling.necessary}.

\subsection{Polynomial-Time Mechanisms}

We next study the power of poly-time approximation mechanisms for maximizing social welfare in \acron{} outcomes (compared to the optimal welfare that can be achieved by any mechanism, poly-time or not, \acron{} or not).
Of particular interest is the question whether the \acron{} requirement entails an additional loss on top of the loss incurred due to the poly-time requirement alone.
In our analysis, we distinguish between mechanisms that operate in the value-query and demand-query models, as is standard in the literature.
We find that in both models, it is possible to construct a \acron{} mechanism with an approximation factor matching that of the best-known approximation algorithms for welfare maximization.
In particular, there exists a poly-time \acron{} mechanism that achieves an $O(\sqrt{m})$ approximation under the demand-query model, and a poly-time \acron{} mechanism that achieves an $O(m/\sqrt{\log m})$ approximation under the value-query model.

We first present a \acron{} approximation mechanism for superadditive valuations using demand queries.  This mechanism, which we call SuperAdditive\acron{} is listed as Algorithm~\ref{super-additive-CWE}.
proceeds in two phases.
In the first phase, it builds a preliminary solution by repeatedly allocating the set that maximizes \emph{value density}.  That is, set $S$ is allocated to agent $i$ so that $\vali(S)/|S|$ is maximized, and this process is then iterated on the remaining items.  A bidder can be allocated to multiple times in this phase, in which case she is allocated the union of the assigned sets.  After all objects have been allocated in this manner, we check whether the welfare can be improved by allocating all objects to a single player; if so, we do so and the mechanism ends.  Otherwise we proceed to phase $2$, where we repeatedly apply local improvements to the allocation.  Specifically, if we write $(x_1, \dotsc, x_n)$ for the tentative allocation, we look for circumstances in which some player $i$ has more value for a set of bundles from among $\{x_1, \dotsc, x_n\}$ than the players to whom those bundles were previously assigned; if such a case exists, we bundle all of these items together and reallocate them to player $i$, then repeat the process with this updated tentative allocation.  Note that this step amounts to repeatedly satisfying the demand of a player in the market with items $\{x_1, \dotsc, x_n\}$ and prices $p_i = v_i(x_i)$, until no further demands are made (which must occur since these prices only increase).  When this process terminates we return the resulting allocation. The following theorem establishes the $O(\sqrt{m})$ approximation and polynomial run time of the algorithm. 

\begin{AvoidOverfullParagraph}
\begin{theorem}
\label{thm.super-additive-alg}
Algorithm SuperAdditive\acron{}
returns a \acron{} outcome that $O(\sqrt{m})$-approximates the optimal social welfare over all assignments.  Furthermore, it can be implemented in a polynomial number of demand queries.
\end{theorem}
\end{AvoidOverfullParagraph}

\begin{algorithm}
\caption{SuperAdditive\acron{}} \label{super-additive-CWE}
\begin{algorithmic}[1]
\REQUIRE Additive valuations $v_1, \dotsc, v_n$.
	\STATE Initialize $R \leftarrow M$, $x_i \leftarrow \emptyset$ for each $i \in N$
\\ \vspace{2mm} \emph{// Phase 1: Build initial allocation}
	\REPEAT
		\STATE $(i,S) \leftarrow \argmax_{i \in [n], S \subseteq R} \{ \vali(S) / |S| \}$
		\STATE $x_i \leftarrow x_i \cup S$, $R \leftarrow R \backslash S$;
	\UNTIL{$R = \emptyset$}
	\STATE {\bf if} $\exists i \in [n]$ such that $\vali(M) > \sum_j \val_j(x_j)$ {\bf then} $x_i \leftarrow M$ and $x_j \leftarrow \emptyset$ for all $j \neq i$;
\\ \vspace{2mm} \emph{// Phase 2: Local search}
	\WHILE{$\exists i \in [n], T \subseteq N$ such that $\vali(\cup_{j \in T} x_j) > \sum_{j \in T} \val_j(x_j)$}
		\STATE Choose $i\in [n]$ and $T \subseteq N$ maximizing $\vali(\cup_{j \in T} x_j) - \sum_{j \in T} \val_j(x_j)$.
		\STATE $x_i \leftarrow x_i \cup (\cup_{j \in T}x_j)$, \ $x_j \leftarrow \emptyset$ for all $j \in T \backslash \{i\}$
	\ENDWHILE
	\RETURN $\allocs$
\end{algorithmic}
\end{algorithm}

To prove Theorem \ref{thm.super-additive-alg}, we first note that the loop on lines 11-14 of Algorithm \ref{super-additive-CWE} must halt after polynomially many iterations.  The reason is that, on each interation of the loop, the number of players $j$ such that $x_j = \emptyset$ increases by at least one, except in the case where $|T| = 1$ and $x_i = \emptyset$.  Consider a sequence of iterations of the loop during which the number of players with empty allocations does not increase. In each iteration, since we choose the $i$ and $T = \{j\}$ that maximizes $\vali(x_j) - \val_j(x_j)$, a certain player $i$ can be chosen only once during this sequence.  This sequence of iterations therefore has length at most $n$.  Since the number of players with empty allocations can increase at most $n$ times, we conclude that this loop must halt after at most $n^2$ iterations.

We next analyze the number of demand queries needed to implement Algorithm \ref{super-additive-CWE}.

\begin{lemma}
Algorithm \ref{super-additive-CWE} can be implemented in a polynomial number of demand queries.
\end{lemma}
\begin{proof}
It is clear that all operations other than lines 3, 11, and 12 can be implemented in polynomial time, so we will focus on those.  Line 3 can be implemented with $n$ \emph{Relative-demand query} in the sense of Blumrosen and Nisan \cite{Blumrosen2005}.  A relative demand query involves generating a vector $\prices$ of non-zero object prices, and returns the set $S$ maximizing $v_i(S) / \sum_{i \in S}p_i$.  Line $3$ can be implemented by posing a demand query to each bidder with the price vector that assigns a price of $1$ to each item in $R$ and an arbitrarily large price to each item in $M \backslash R$.  Since a relative demand query can be implemented by a polynomial number of demand queries \cite{Blumrosen2005}, Line 3 can be implemented in polynomially many demand queries as well.

Lines 11 and 12 can be implemented as follows.  We focus on line 12; line 11 can be implemented in the same way.  We will define a new auction instance, over the original player set $N$ and a new set of objects $\{a_1', \dotsc, a_n'\}$.  Object $a_i'$ is intended to represent the set $x_i$ in the original auction.  Given some $T \subseteq [n]$, the value of player $i$ for some object set $x' = \{ a_j' : j \in T \}$ will be defined as $\vali'(x') := \vali( \cup_{j \in T} x_j )$.  Note that this new auction corresponds precisely to the original auction with the objects partitioned into indivisible bundles $x_1, \dotsc, x_n$.  We can then implement line 12 by making a single demand query to each agent $i$, \emph{with respect to this new bundled auction}, under prices given by $p_j' = \val_j(x_j)$.  If any player has positive value for the demanded set returned by this query, then this player and his demanded set will satisfy the condition on line 12.  Whichever player has the highest positive value, plus his demanded set, will be one chosen on line 12.  If every player has value $0$ for his demanded set, then the condition on line 9 evaluates to false.
\end{proof}

To show correctness, we first show that SuperAdditive\acron{} generates a \acron{} outcome, then bound the social welfare of the outcome it returns.

\begin{claim}
Algorithm SuperAdditive\acron{} generates a \acron{} outcome.
\end{claim}
\begin{proof}
The outcome $\allocs$ generated by SuperAdditive\acron{} is such that, for all $i \in [n]$ and all $T \subseteq N$, $\vali(\cup_{j \in T} x_j) \leq \sum_{j \in T} \val_j(x_j)$.  Thus, if we set prices $\pricei = \vali(x_i)$, we have that $\vali(x_i) - p_i = 0 \geq \vali(\cup_{j \in T} x_j) - \sum_{j \in T} p_j$ for all $i$ and all $T \subseteq N$.  Thus outcome $(\allocs, \prices)$ is \acron{}.
\end{proof}

\begin{claim}
Algorithm SuperAdditive\acron{} obtains an $O(\sqrt{m})$ approximation to the optimal social welfare.
\end{claim}
\begin{proof}
We will argue that, on any superadditive input profile, the social welfare generated by our mechanism is at least that of a mechanism proposed by Blumrosen and Nisan, which was shown to obtain an $O(\sqrt{m})$ approximation to the optimal social welfare \cite{Blumrosen2005}.  The mechanism due to Blumrosen and Nisan, which we will refer to as the BN mechanism, is identical to the first phase of our mechanism, except that it only allows a single set to be allocated to any single player in the loop on lines 2-5.

Consider first the case in which all agents are single-minded.  In this case, phase 1 of our mechanism is identical to the BN mechanism (since our mechanism can allocate to any agent at most once in lines 2-6, by single-mindedness).  Thus, after phase 1, the preliminary allocation generated by our mechanism has social welfare equal to that of the BN mechanism.  We then note that the social welfare of the allocation can only increase over each iteration of the loop at lines 11-15.  This is because, whenever the condition on line 11 evaluates to true, it must be that $x_i = \emptyset$ (again, due to single-mindedness), so the increase in value to agent $i$ is greater than the decrease in value to each agent $j \in T$.  We conclude that the social welfare of the final allocation is at least that of the allocation returned by the BN mechanism.

We now consider the general case of superadditive valuations.  We will compare our auction with a corresponding auction in which each original superadditive bidder is represented by $2^m$ single-minded bidders, one for each subset of $M$.  That is, for each $i \in [n]$ and $S \subseteq M$, we will have a bidder $(i,S)$ with a single-minded value of $\vali(S)$ for set $S$.  Note first that the optimal social welfare for this new auction is identical to that of the original auction.  Also, since this is an auction over single-minded bidders, our mechanism obtains at least as much social welfare as the BN auction for this new auction instance (as argued above).  We will now argue that our mechanism obtains at least as much welfare in the original auction as in the new auction instance; this will complete the proof.


The only difference in the mechanism's behavior on these two input instances is that, on line 4 or 13, we may allocate multiple times to a single agent in the original input instance, whereas in the expanded auction instance these allocations are to separate bidders.  Note, however, that if we allocate multiple sets $S_1, \dotsc, S_k$ to a player $i$ in the original instance, then we allocate precisely to players $(i, S_1), \dotsc, (i, S_k)$ in the expanded instance.  Since bidder valuations are superadditive, the value of allocating multiple disjoint sets to a single bidder cannot be less than the value of allocating these sets separately to single-minded bidders with the same valuations.  We conclude that the welfare generated by our mechanism in the original auction is at least that of the BN mechanism in the corresponding auction with single-minded types, and hence obtains an $O(\sqrt{m})$ approximation to social welfare as required.
\end{proof}

We next move to the value-query model. We show that the $O(m/\sqrt{\log m})$ approximation mechanism due to \cite{HKNS2004} is guaranteed to generate CWE outcomes.
We note that it nearly matches the lower bound of $O(m / \log m)$ on the approximability of CAs with superadditive bidders (using value queries) \cite{Mirrokni2008}.

\begin{theorem}
If agents are super-additive, then there exists a mechanism that makes a polynomial number of value queries and generates a CWE outcome that achieves a $O(m/\sqrt{\log m})$ approximation to the optimal social welfare.
\end{theorem}

\begin{proof}
The mechanism groups the objects into $\log m$ bundles, each of size $m/\log{m}$, arbitrarily.  It then returns the bundle-efficient allocation over those bundles.  This is known to achieve a $O(m/\sqrt{\log m})$ approximation, and can be implemented in a polynomial number of value queries \cite{HKNS2004}.  Since the allocation is bundle-efficient, it is CWE.
\end{proof}

\subsection{Single-minded Valuations and Value Queries}
\label{sec.single.minded}
In the special case in which agents are single-minded, Algorithm \ref{super-additive-CWE} can be improved to run in a polynomial number of \emph{value} queries, obtaining an $O(\sqrt{m})$ approximation to the optimal welfare.
This mechanism is new; as far as we are aware, existing $O(\sqrt{m})$-approximation mechanisms do not satisfy \acron{}.

%
Our algorithm, which we call SingleMinded\acron{} and listed as Algorithm~\ref{single-minded-CWE}, proceeds as follows.
We split the bidders into two groups: those with ``small'' desired sets (of size at most $\sqrt{m}$) and those with larger desired sets.  We first generate a provisional allocation that only includes the bidders with small desired sets.  We construct this preliminary allocation greedily: we order players from largest value to smallest, then allocate to players in this order if their desired set is available.  Any object that is left unallocated is then added to an arbitrary \emph{non-empty} allocation.  Then, in the second phase of the algorithm, we consider those bidders with large desired sets.  We order these large-set bidders from highest value to smallest, and for each bidder $i$ in this order, say with value $v_i$ for set $S_i$, we ask whether $v_i$ is greater than the sum of values of all players whose allocations intersect set $S_i$.  If so, we take all of those intersecting allocations from their respective bidders and allocate them all to player $i$.  After this operation has been completed for every large-set bidder (in order from highest value to smallest), we return the resulting allocation.
The following theorem establishes the $O(\sqrt{m})$ approximation and polynomial run time of the algorithm. 

\begin{algorithm}
\caption{SingleMinded\acron{}} \label{single-minded-CWE}
\begin{algorithmic}[1]
\REQUIRE Single-minded valuations $(v_1, S_1), \dotsc, (v_n, S_n)$.
\\ \vspace{2mm} \emph{// Phase 1: Players with small desired sets}
	\STATE $Q \leftarrow \{ i \in N : |S_i| \leq \sqrt{m} \}$
	\STATE $\allocs \leftarrow $ outcome of greedily allocating sets to players in $Q$, by value.
	\FOR{each object $a_i$ not allocated in $\allocs$}
		\STATE add $a_i$ to an arbitrary non-empty element of $\allocs$
	\ENDFOR
\\ \vspace{2mm} \emph{// Phase 2: Players with large desired sets}
	\FOR{$i \in N - Q$, in decreasing order of $v_i$}
		\STATE $C_i \leftarrow \{ j : x_j \cap S_i \neq \emptyset \}$
		\STATE {\bf if} $v_i > \sum_{j \in C_i} v_j$ {\bf then} $x_i \leftarrow \cup_{j \in C_i} x_j$, \ $x_j \leftarrow \emptyset$ for each $j \in C_i$;
	\ENDFOR
	\RETURN $\allocs$
\end{algorithmic}
\end{algorithm}

\begin{theorem}
\label{thm.single.minded.cwe}
When agent valuations are single-minded,
Algorithm SingleMinded\acron{}
returns a \acron{} outcome that $O(\sqrt{m})$-approximates the optimal social welfare over all assignments.  Furthermore, it can be implemented in a polynomial number of value queries.
\end{theorem}

The runtime requirement of Theorem \ref{thm.single.minded.cwe} is immediate from the definition of Algorithm \ref{single-minded-CWE}.  We address the approximation and \acron{} requirements separately.

\begin{lemma}
Algorithm \ref{single-minded-CWE} obtains an $O(\sqrt{m})$ approximation when agents are single-minded.
\end{lemma}
\begin{proof}
We first note that, upon each iteration of loop $6-12$, the social welfare of allocation $\allocs$ cannot decrease.  This is because, on each iteration, the allocation is altered precisely when the value of player $i$ for set $\cup_{j \in C_i} x_j$ is greater than the value being obtained from these objects in the pre-existing allocation.

After line $2$, the allocation $\allocs$ corresponds to the outcome of the greedy allocation among bids for sets of size at most $\sqrt{m}$.  Thus, since social welfare does not decrease from this point onward, the final allocation generates at least the welfare of this greedy outcome.

Next, suppose that $i$ is the player for which $v_i$ is maximized.  We will show that our algorithm generates social welfare at least $v_i$.  If $i \in Q$, then we will have $S_i \subseteq x_i$ after line $2$, so our algorithm's welfare must be at least $v_i$.  Otherwise, $i$ will be the first agent considered in the loop starting at line $3$.  If $v_i$ is larger than the value generated by our algorithm before this first iteration, then the condition on line $8$ must evaluate true and we will have $S_i \subseteq x_i$ after the first iteration of the loop.  Thus, again, our algorithm's welfare must therefore be at least $v_i$.  In either case, our mechanism obtains welfare at least $\max_i v_i$.

We therefore have that our mechanism generates at least as much welfare as an algorithm that chooses the welfare-maximizing outcome between (a) the greedy allocation among bids for sets of size at most $\sqrt{m}$, and (b) allocating all items to a single bidder.  However, it is known that such an allocation obtains an $O(\sqrt{m})$ approximation to the optimal social welfare \cite{MN-08}.  Since our mechanism generates at least this much welfare, it must be an $O(\sqrt{m})$ approximation as well.
\end{proof}

\begin{lemma}
Let $\allocs$ be the allocation returned by Algorithm \ref{single-minded-CWE}, and consider price vector $\pricei = \vali(\alloci)$.  Then $(\allocs, \prices)$ is \acron{}.
\end{lemma}
\begin{proof}
For a given fixed input profile, let $Q = \{i \in [n] : |S_i| \leq \sqrt{m} \}$ be the set of agents with ``small'' demand sets and let $\allocs$ be allocation chosen by the mechanism.  For each player $i$ such that $\alloci = \emptyset$, let $C_i = \{ j : x_j \cap S_i \neq \emptyset \}$.  It suffices to show that, for each $i$ such that $\alloci = \emptyset$, we have $\vali \leq \sum_{j \in C_i} v_j$.  (Since $S_i \subseteq x_i$ whenever $x_i \neq \emptyset$).

Observe that, after line 5, the allocation $x_j$ forms a partition of the set of objects.  We also note that, for agent $i$, the value of $\sum_{j \in C_i} v_j$ can only increase over each iteration of the loop on lines 4-8.  This is because, for each object $a_i$, the value of the player $j$ such that $a_i \in x_j$ can only increase (since the allocation only changes when a set of objects is transferred to a player with value at least that of all previous owners of those objects).

Let $i$ be some player such that $\alloci = \emptyset$.  Suppose $i \in Q$, and that $i$ did not receive his desired set from the greedy allocation on line $2$.  Then, since $\alloci = \emptyset$, it must be that $\vali \leq \max\{ j \in Q, j \neq i, S_j \cap S_i \neq \emptyset : v_j\}$ (from the definition of the greedy algorithm).  Since the values of winners only increase over the course of the mechanism, we must have $\vali \leq \max\{ j, x_j \cap S_i \neq \emptyset : v_j\}$ and hence $\vali \leq \sum_{j \in C_i} v_j$.

Next suppose that $i \not\in Q$, and the condition on line 6 does not evaluate to true when agent $i$ is considered on lines 4-8.  Then, at this point in the algorithm we have $\vali \leq \sum_{j \in C_i} v_j$, and since winner values only increase over the course of the mechanism this will be true at the completion of the mechanism.

Finally, suppose that $i$ is such that, on some iteration of the loop on lines 6-12, $\alloci$ is initially non-empty but is set to $\emptyset$ over the course of the iteration.  Then the condition on line $8$ evaluated true for some $j$ such that $x_i \cap S_j \neq \emptyset$, and hence after this iteration we have $v_i \leq v_j$ and $S_i \cap x_j \neq \emptyset$.  Since winner values only increase over the course of the mechanism, we must have $\vali \leq \sum_{j \in C_i} v_j$ at the completion of the mechanism, as required.
\end{proof}

%% file: budget-additive.tex
\section{Uniform Budget-Additive Valuations}
\label{sec:budget-additive}

A budget-additive valuation is specified by budget $B$ and item values $v_j$, $j \in M$.  The value of set $S$ is then $v(S) = \min\{B, \sum_{j \in S} v_j\}$.
Note that the uniform-valuation case implies $v_i(j) \leq B_i$ for all $i \in [n], j\in M$.

The problem of maximizing social welfare with budget-additive valuations has been extensively studied in recent years from a computational perspective \cite{Chakrabarty2008,Srinivasan2008,Feige2006,Andelman04,Azar08}.
It is known to be APX-hard \cite{Chakrabarty2008}, and the best known approximation ratio is $4/3$ (achieved via iterative rounding or primal-dual algorithms).
This factor also matches the integrality gap of the corresponding linear program.
In this section we consider problem instances in which agents have \emph{uniform} budget-additive valuations \cite{Andelman04}, in the sense that for each object $j$ there is a value $v_j$ such that, for all $i$, $v_i(j) \in \{0 , v_j\}$.  In other words, each object has a fixed value $v_j$; each player with non-zero value for $j$ values it at $v_j$.
Here too, the best approximation known is $4/3$.

\subsection{Arbitrary budgets}

\begin{figure}[t]
\centering
\subfigure[An instance of uniform budget additive bidders that admits a gap in welfare due to \acron{}.]
{
\label{fig:sub1}
\input{ab-gap2}
}
\hspace{0.3cm}
\subfigure[An instance of uniform budget additive bidders with identical budgets that admits no Walrasian equilibrium.]
{
\label{fig:sub2}
\input{ab-gap3}
}
\caption{Instances of uniform budget-additive agents. Agents and items are represented by thin and thick nodes, respectively. An edge between an agent and item nodes means that the agent values the item. Values written next to agents and items correspond to budgets and (uniform) values, respectively.}
\label{fig:ab-gap}
\end{figure}
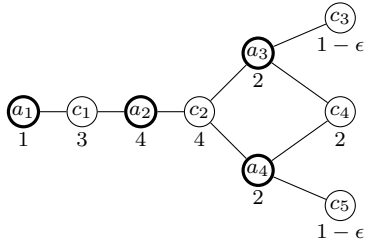
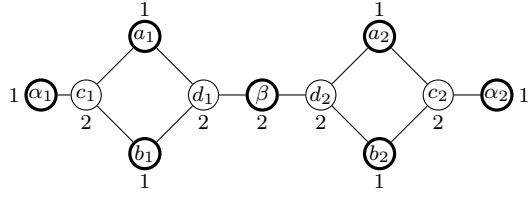

We first give an example in which no \acron{} allocation can achieve more than a $7/8$ fraction of the optimal welfare.

\begin{claim}
There is a profile of uniform budget-additive valuations for which no \acron{} allocation achieves more than $7/8$ of the optimal welfare.
\end{claim}
\begin{proof}
Consider the instance given in Figure~\ref{fig:sub1} with 5 players, $\{c_i\}_{i=1..5}$ (thin nodes), and 4 items, $\{a_i\}_{i=1..4}$ (thick nodes).
The value written next to an agent node (resp., item node) corresponds to agent's budget (resp., item's value).
If an agent values some item, we draw an edge between the agent and item.

In this example, the optimal fractional assignment has social welfare $8$ (for example, the following achieves welfare $8$:
$c_1$ gets sets $\{a_1\}$ and $\{a_2\}$ with probability $1/2$ each,
$c_2$ gets sets $\{a_2\}$ and $\{a_3,a_4\}$ with probability $1/2$ each, and
$c_4$ gets sets $\{a_3\}$ and $\{a_4\}$ with probability $1/2$ each).
However, the optimal integral assignment obtains value $8-\epsilon$ (for example, in one optimal assignment $c_1$ gets $\{a_1\}$, $c_2$ gets $\{a_2\}$, $c_4$ gets $\{a_4\}$, and $c_3$ gets $\{a_3\}$).
Therefore, in order to get a \acron{} outcome, one must bundle some items together.
We claim that, for any bundling choice, one cannot achieve a higher welfare than $7$.  The argument proceeds via a case analysis:
\begin{enumerate}
\item  If the bundle is $\{a_1, a_2\}$, then if it is allocated to $c_1$ its contribution is $3$, and the total value of other items is $4$, for a total of $7$.  If $\{a_1, a_2\}$ is allocated to $c_2$ then one of $a_3$ or $a_4$ must be allocated to a player of budget $1 - \epsilon$, and hence the total value is at most $4 + 2 + 1-\epsilon < 7$.
\item  If the bundle is $\{a_2, a_3\}$ or $\{a_2, a_4\}$, then the bundle can contribute at most $4$ to the total welfare, and the total value of other items is $3$, for a total of $7$.  
\item  If the bundle is $\{a_3, a_4\}$, then if it is allocated to $c_4$ its contribution is $2$ and the total value of other items is $5$ for a total of $7$.  If it is allocated to $c_2$ then the only other player who can obtain value is $c_1$ with a budget of $3$, so the total value is at most $4+3 = 7$.
\item  If the bundle is $\{a_1, a_3\}$ or $\{a_1, a_4\}$, say $\{a_1, a_3\}$, then if the bundle is allocated to $c_1$ or $c_3$ then its value is at most $1$, which leads to total value at most $7$ since the other items have value $6$.  If the bundle is allocated to $c_2$ (for a value of $2$) then $a_2$ can contribute value at most $3$ and $a_4$ can contribute value at most $2$, for a total of $7$.  Finally, if the bundle is allocated to $c_4$ then whichever of the remaining items is not allocated to $c_2$ will contribute one less than its value to the total, and hence the total value is at most $2 + 2+ 4 - 1 = 7$.
\end{enumerate}
We conclude that every possible bundling and allocation results in a welfare of at most $7$, and hence the gap is established.
\end{proof}

%

We next present an algorithm that converts any allocation into a \acron{} allocation while preserving at least half of the original welfare.  The algorithm, presented as Algorithm \ref{uniform-budget-additive-CWE}, proceeds as follows.  Given an allocation profile $\allocs = (x_1, \dotsc, x_n)$, the algorithm checks whether there is an instance in which $v_j(x_i) > v_i(x_i)$.  Note that this can occur only if $\sum_{k \in x_i} v_k > B_i$ and $B_j > B_i$.  If there is no such instance, then the algorithm terminates and the current allocation is returned.  Otherwise, if there is a pair of agents $i$ and $j$ with $v_j(x_i) > v_i(x_i)$, then the item of lowest value in $x_i$ is removed from $x_i$ and added to $x_j$.  The algorithm then repeats.  Note that the algorithm must terminate, since every iteration results in an item being shifted from one agent to another agent with strictly larger budget.

The intuition behind UniformBudgetAdditive\acron{} is that we would like to reduce the social welfare of the optimal \emph{fractional} allocation of the bundles in $\mathbf{x}$.  Indeed, if there are no instances in which $v_j(x_i) > v_i(x_i)$, then it must be that the assignment $\mathbf{x}$ is an optimal \emph{fractional} allocation, and hence by Claim \ref{claim.frac} is \acron{}.  Thus, by transforming the input allocation into an allocation that satisfies this property, we are reducing the integrality gap of our reduced market to $1$.  What remains to show is that this iterative procedure does not drastically reduce the welfare of the allocation.  

\begin{theorem}
\label{thm.ba.non-identical}
Suppose $\allocs$ is an arbitrary allocation with welfare $SW(\allocs)$. Given $\allocs$, Algorithm UniformBudgetAdditive\acron{} returns, in polynomial time, a \acron{} outcome $\allocs'$ such that $SW(\allocs') \geq \frac{1}{2} SW(\allocs)$.
\end{theorem}

\begin{proof}
First, note that we can assume without loss of generality that $\allocs'$ is such that $j \in x'_i$ implies $v_i(j) > 0$ (else, take any object $j$ such that $j \in x'_i$ and $v_i(j) = 0$ and re-allocate it to an arbitrary agent with non-zero value for it; this cannot reduce welfare).

We next prove that $v_i(x'_i) = \max_{\ell \in [n]}v_{\ell}(\alloc'_i)$ for every $i \in [n]$.
If $\sum_{j \in x'_i}v_j$ does not exceed $B_i$, then $v_i(x'_i)$ is clearly the largest possible value of $x'_i$ by any agent.
If $\sum_{j \in x'_i}v_j$ exceeds $B_i$, then $v_i(x'_i)=B_i$.
Assume toward contradiction that there exists an agent $i'$ such that $v_{i'}(\alloc'_i) > v_i(\alloc'_i)$.
This implies that $B_{i'} > B_i$ and there exists an item $j \in x'_i$ such that $v_{i'}(j) > 0$.
But if this were the case, item $j$ would have moved to an agent with higher budget in line 10 of the algorithm; reaching a contradiction.
It follows that the optimal fractional allocation of bundles $\allocs'$ achieves the same social welfare of allocation $\allocs'$, and hence $\allocs'$ is CWE.

It remains to prove that $SW(\allocs') \geq \frac{1}{2} SW(\allocs)$.
Since the mechanism moves items in a non-decreasing order, and $v_i(j) \leq B_i$ for every $i,j$, it is guaranteed that every agent $i$ that exceeds his budget in $\allocs$, is left with value of at least $B_i/2$ in $\allocs'$.
The desired approximation follows.
\end{proof}

\begin{algorithm}
\caption{UniformBudgetAdditive\acron{}} \label{uniform-budget-additive-CWE}
\begin{algorithmic}[1]
\REQUIRE Uniform budget additive valuations $v_1, \dotsc, v_n$, agent budgets $B_1, \ldots, B_n$, and an allocation profile
$\allocs=(\alloc_1, \ldots, \alloc_n)$ (w.l.o.g., every item is allocated to a player with positive value for it).
\STATE Initialize $S \leftarrow \emptyset$
\IF{$\sum_{k \in x_i} v_i(k) \leq B_i$ for every $i$}
\RETURN $\allocs$
\ELSE
\FOR{$i \in [n]$, in a non-decreasing order of $B_i$}
\STATE $S \leftarrow \{ j \in x_i : \exists i' \mbox{ s.t. } v_{i'}(j)>0 \mbox{ and } B_{i'} > B_i\}$
\WHILE{$\sum_{k \in x_i} v_i(k) > B_i$ and $S \neq \emptyset$}
\STATE $j \leftarrow \argmin_{j \in S} \{ v_j \}$
\STATE $i' \leftarrow \argmax_{i \in [n]} \{B_i : v_i(j) > 0 \}$
\STATE $x_i \leftarrow x_i - j,\ S \leftarrow S - j,\ x_{i'} \leftarrow x_{i'} \cup j$ (i.e., move $j$ from $x_i$ to $x_{i'}$)
\ENDWHILE
\ENDFOR
\RETURN $\allocs$
\ENDIF
\end{algorithmic}
\end{algorithm}

\begin{AvoidOverfullParagraph}
As a corollary, the known $4/3$ approximation algorithm for budget-additive valuations can be turned into an $8/3$ approximation \acron{} mechanism in the case of uniform values.
Furthermore, the allocation returned by
Algorithm UniformBudgetAdditive\acron{}
has the property that the full surplus can be extracted from the \acron{} outcome as revenue.
\end{AvoidOverfullParagraph}

\begin{AvoidOverfullParagraph}
\begin{lemma}
\label{lem.ba.non-identical.revenue}
If agents are sub-additive and $\allocs$ is an allocation such that $v_i(\alloc_i) = \max_{\ell \in [n]} v_{\ell}(\alloc_i)$ for every $i \in [n]$, then the price vector $p_i = v_i(\alloc_i)$ is such that $(\allocs,p)$ is \acron{}.
\end{lemma}
\end{AvoidOverfullParagraph}

\begin{proof}
Pick any agent $i$ and set of agents $S$.
Since for every $i$, $v_i(\alloc_i) = \max_{\ell \in [n]} v_{\ell}(\alloc_i)$, it follows that
$\sum_{j \in S} v_i(\allocs_j) - \sum_{j \in S} v_j(\allocs_j) \leq 0$.
We get:
\begin{eqnarray*}
v_i(\alloc_i) - p_i =  0 \geq  \sum_{j \in S} v_i(\allocs_j) - \sum_{j \in S} v_j(\allocs_j)
& = & \sum_{j \in S} v_i(\allocs_j) - \sum_{j \in S} p_j \\
& \geq & v_i (\cup_{j \in S} \allocs_j) - \sum_{j \in S}p_j,
\end{eqnarray*}
where the last inequality follows from sub-additivity.
The assertion of the lemma follows.
\end{proof}

As a corollary, we get the following approximation result.

\begin{corollary}
\label{cor.ba.non-identical.revenue}
For every instance of uniform budget-additive valuations, there exists a poly-time mechanism that generates a \acron{} outcome that is within factor $8/3$ of the optimal revenue.
\end{corollary}

\begin{proof}
We have shown that one can turn the 4/3 SW-approximation algorithm for budget-additive valuations into an 8/3 SW-approximation CWE mechanism that has the property $v_i(\alloc_i) = \max_{\ell \in [n]} v_{\ell}(\alloc_i)$ for every $i \in [n]$.
The last lemma implies that the full surplus can be extracted by a CWE mechanism.
Observing that the social welfare is an upper bound on the profit establishes the assertion of the corollary.
\end{proof}

\subsection{Identical Budgets}

In this section we study the restricted case of identical budgets.
We note that the welfare maximization problem is NP-hard, even under this restriction, as it includes PARTITION as a special case\footnote{For a given instance of PARTITION with integers $a_1, \ldots, a_n$ such that $\sum_{j \in [n]}a_j=2B$, construct an instance of our problem with two agents, each with budget $B$, and $n$ items, with item $j$ having value $a_j$.}.
Moreover, we shall show that there are input instances in this class for which Walrasian equilibria do not exist\footnote{While this is well known for the case of arbitrary budget additive valuations, here we consider the further restricted class of uniform values and identical budgets.}.
Nevertheless, we shall show that any allocation can be transformed into a \acron{} allocation with no loss to the social welfare.

We first show that there exist input instances for which no Walrasian equilibrium exists.
Consider the instance given in Figure~\ref{fig:sub2} with 4 players, $c_i, d_i$ for $i \in \{1,2\}$, and 7 items: $a_i, b_i, \alpha_i$, and $\beta$ for $i \in \{1,2\}$.  Each player has budget $2$.  Each item has value $1$, except for item $\beta$ which has value $2$.
For $i \in \{1,2\}$, player $c_i$ has value for objects $a_i, b_i$, and $\alpha_i$, and player $d_i$ has value for objects $a_i, b_i$, and $\beta$.
In this example, the optimal fractional assignment has social welfare $8$.  For example, the fractional assignment in which $c_i$ gets sets $\{a_i, \alpha_i\}$ and $\{b_i, \alpha_i\}$ with probability $1/2$, and $d_i$ gets sets $\{a_i, b_i\}$ and $\{\beta\}$ with probability $1/2$, for $i \in \{1,2\}$, achieves welfare $8$.  However, the optimal integral assignment obtains value at most $7$.  Thus the optimal fractional welfare does not occur at an integral solution, and hence a WE does not exist.

We now show that, for any allocation, there exists a \acron{} allocation that obtains at least as much social welfare, and moreover this \acron{} allocation can be found efficiently.  

\begin{theorem}
\label{thm.identical.sw}
Suppose $\allocs$ is an arbitrary allocation.  There exists \acron{} allocation $\allocs'$ with $SW(\allocs') \geq SW(\allocs)$, and $\allocs'$ can be found in polynomial time given $\allocs$.
\end{theorem}

\begin{proof}
Given $\allocs$, we construct $\allocs'$ by taking any object $j$ such that $j \in x_i$ with $v_i(j) = 0$ and re-allocating it to an arbitrary agent with non-zero value for it.  This can be done in polytime and can only increase the social welfare of the resulting allocation.  It remains to show that this allocation is \acron{}.

Suppose $\allocs'$ is such that $j \in x_i'$ implies $v_i(j) > 0$.  In this case, $v_i(x_i') = \min\{B, \sum_{j \in x_i'} v_j\} = \max_\ell v_\ell(x_i')$ for each $i$.  The social welfare of allocation $\allocs'$ is therefore $\sum_i \max_{\ell} v_\ell(x_i')$, which is an upper bound on the value of any fractional assignment of bundles $x_1', \dotsc, x_n'$.  Thus the optimal fractional allocation of these bundles obtains the same welfare as $\allocs'$, and hence $\allocs'$ is \acron{} as required.
\end{proof}

As a corollary, the known $4/3$ approximation algorithms for budget-additive valuations can be made \acron{} for the case of uniform values and identical budgets.
Note that the complementary case, in which the budgets are identical but the item values are non-uniform, also has instances where no \acron{} allocation can achieve the welfare of an optimal allocation.
Consider, for example, an instance with 3 players, $c_1, c_2, c_3$, each with budget $2$, and 3 items, $a_1,a_2,a_3$, where $v_{c_1}(a_1)=2, v_{c_1}(a_2)=v_{c_1}(a_3)=1$, $v_{c_2}(a_2)=v_{c_2}(a_3)=2$, and $v_{c_3}(a_1)=4 \epsilon, v_{c_3}(a_3)=\epsilon$ (and all other values are zero).
One can easily verify that there exists a fractional solution that achieves welfare $4+2\epsilon$, while the highest welfare that can be obtained by any integral solution is $4+\epsilon$.
However, by bundling together any two items, the highest achievable welfare is reduced to at most $4$.

%% file: ab-gap2.tex
\begin{tikzpicture}[scale=1.2]

\tikzstyle{dot}=[circle,draw=black,fill=white,thin,inner sep=0pt,minimum size=4mm]
\tikzstyle{point}=[circle,draw=black,fill=white,very thick,inner sep=0pt,minimum size=4mm]

\node (a1) at (0,1*1.3) [point] {\scriptsize{$a_1$}};
\node at (0,1.3-0.3) {\scriptsize{$1$}};

\node (c1) at (.5*1.3,1*1.3) [dot] {\scriptsize{$c_1$}};
\node at (.5*1.3,1*1.3-0.3) {\scriptsize{$3$}};

\node (a2) at (1*1.3,1*1.3) [point] {\scriptsize{$a_2$}};
\node at (1*1.3,1*1.3-0.3) {\scriptsize{$4$}};

\node (c2) at (1.5*1.3,1*1.3) [dot] {\scriptsize{$c_2$}};
\node at (1.5*1.3,1*1.3-0.3) {\scriptsize{$4$}};

\node (a3) at (2*1.3,1.5*1.3) [point] {\scriptsize{$a_3$}};
\node at (2*1.3,1.5*1.3-0.3) {\scriptsize{$2$}};

\node (a4) at (2*1.3,0.5*1.3) [point] {\scriptsize{$a_4$}};
\node at (2*1.3,0.5*1.3-0.3) {\scriptsize{$2$}};

\node (c3) at (2.7*1.3,1.8*1.3) [dot] {\scriptsize{$c_3$}};
\node at (2.7*1.3,1.8*1.3-0.3) {\scriptsize{$1-\epsilon$}};

\node (c4) at (2.7*1.3,1*1.3) [dot] {\scriptsize{$c_4$}};
\node at (2.7*1.3,1*1.3-0.3) {\scriptsize{$2$}};

\node (c5) at (2.7*1.3,0.2*1.3) [dot] {\scriptsize{$c_5$}};
\node at (2.7*1.3,0.2*1.3-0.3) {\scriptsize{$1-\epsilon$}};

%
%
%
%
%
%
%
%
%

\path (c1) edge node [above] {} (a1);
\path (c1) edge node [above] {} (a2);
\path (c2) edge node [above] {} (a2);
\path (c2) edge node [above] {} (a3);
\path (c2) edge node [above] {} (a4);
\path (c3) edge node [above] {} (a3);
\path (c4) edge node [above] {} (a4);
\path (c4) edge node [above] {} (a3);
\path (c5) edge node [above] {} (a4);

\end{tikzpicture}

%% file: ab-gap3.tex
\begin{tikzpicture}[scale=1.2]

\tikzstyle{dot}=[circle,draw=black,fill=white,thin,inner sep=0pt,minimum size=4mm]
\tikzstyle{point}=[circle,draw=black,fill=white,very thick,inner sep=0pt,minimum size=4mm]

\node (alpha1) at (-.5,1*1.3) [point] {\scriptsize{$\alpha_1$}};
\node at (-0.8,1.3) {\scriptsize{$1$}};

\node (c1) at (0,1*1.3) [dot] {\scriptsize{$c_1$}};
\node at (0,1*1.3-0.3) {\scriptsize{$2$}};

\node (d1) at (1*1.3,1*1.3) [dot] {\scriptsize{$d_1$}};
\node at (1.3,1) {\scriptsize{$2$}};

\node (b1) at (.5*1.3,.5*1.3) [point] {\scriptsize{$b_1$}};
\node at (.5*1.3,.5*1.3-0.3) {\scriptsize{$1$}};

\node (a1) at (.5*1.3,1.5*1.3) [point] {\scriptsize{$a_1$}};
\node at (.5*1.3,1.5*1.3+0.3) {\scriptsize{$1$}};

\node (beta) at (1.5*1.3,1*1.3) [point] {\scriptsize{$\beta$}};
\node at (1.5*1.3,1*1.3-0.3) {\scriptsize{$2$}};

\node (d2) at (2*1.3,1*1.3) [dot] {\scriptsize{$d_2$}};
\node at (2*1.3,1*1.3-0.3) {\scriptsize{$2$}};

\node (a2) at (2.5*1.3,1.5*1.3) [point] {\scriptsize{$a_2$}};
\node at (2.5*1.3,1.5*1.3+0.3) {\scriptsize{$1$}};

\node (c2) at (3*1.3,1*1.3) [dot] {\scriptsize{$c_2$}};
\node at (3*1.3,1*1.3-0.3) {\scriptsize{$2$}};

\node (b2) at (2.5*1.3,1/2*1.3) [point] {\scriptsize{$b_2$}};
\node at (2.5*1.3,1/2*1.3-0.3) {\scriptsize{$1$}};

\node (alpha2) at (3*1.3+0.5*1.3,1*1.3) [point] {\scriptsize{$\alpha_2$}};
\node at (4.85,1.3) {\scriptsize{$1$}};

\path (c1) edge node [above] {} (alpha1);
\path (c1) edge node [above] {} (a1);
\path (c1) edge node [above] {} (b1);
\path (d1) edge node [above] {} (a1);
\path (d1) edge node [above] {} (b1);
\path (d1) edge node [above] {} (beta);
\path (d2) edge node [above] {} (beta);
\path (d2) edge node [above] {} (a2);
\path (d2) edge node [above] {} (b2);
\path (c2) edge node [above] {} (a2);
\path (c2) edge node [above] {} (b2);
\path (c2) edge node [above] {} (alpha2);

\end{tikzpicture}

%% file: appendix.tex

\appendix

\section{Relation to envy-free (EF) notions}
\label{sec:ef}

The \acron{} notion is closely related to several notions of \emph{envy freeness} that have been extensively studied in recent years.
In what follows, we shall describe the difference and similarities between \acron{} and these notions.

\paragraph{Envy free 1 (EF1)} (as in, e.g., \cite{MuAlem2009,Cohen2010,Cohen2011,Fleischer2011,Papai2003}).

Let $\allocs=(\alloc_1, \ldots, \alloc_n)$ be an allocation, and let $\prices=(\price_1, \ldots, \price_n)$ be prices, such that $p_i$ is the price of $\alloc_i$.
An outcome $(\allocs,\prices)$ is said to be EF1
if no agent wishes to switch her outcome with another;
i.e., for every $i,j \in [n]$, $v_i(\alloc_1) - \price(\alloc_i) \geq v_i(\alloc_j) - \price(\alloc_j)$.

An obvious way to view the difference between \acron{} and EF1 is that EF1 requires that no agent envies a single other agent,
while \acron{} precludes envy in any set of other agents.
However, a more subtle examination reveals that any EF1 outcome can be extended to an outcome that exhibits the latter (stronger) requirement,
but only if one allows for arbitrary bundle prices.
This is in stark contrast to the linear prices required by \acron{}.
Moreover, \acron{}, unlike EF1, requires market clearing.

\paragraph{Envy free 2 (EF2)} (as in, e.g., \cite{Guruswami2005,Balcan2008,Cheung2008}).

Let $\allocs=(\alloc_1, \ldots, \alloc_n)$ be an allocation, and let $\prices=(\price_1, \ldots, \price_m)$ be item prices, assigning a price $p_j$ for every $j \in M$.
An outcome $(\allocs,\prices)$ is said to be EF2 if $\alloc_i \in D_i(\prices)$ for every $i \in [n]$.
This notion entails the first requirement of a Walrasian equilibrium, but does not require market clearing.

\acron{} is, on the one hand, a relaxation of this notion, as it allows for a pre-sale bundling phase.
On the other hand, it requires market clearing, whereas EF2 does not.

\section{Revenue enhancement}
\label{app:revenue}

We now point out that revenue can increase substantially when we optimize revenue with respect to equilibria with bundles, as opposed to the individual item prices of standard Walrasian equilibrium.
For example, consider a case with two bidders and three objects $\{a_1,a_2,a_3\}$. Player $1$'s valuation is given by $v_1(S) = 1$ if $a_1 \in S$, otherwise $v_1(S) = 0$.  Choose an arbitrarily large constant $R$.  Then player $2$'s valution is additive up to a capacity of $2$ objects, with values $R-1, R, R$ for items $a_1,a_2,a_3$ respectively.  In this example, both players satisfy gross substitutes, and the efficient allocation is $(\{a_1\}, \{a_2,a_3\})$.  This allocation can be supported by a Walrasian equilibrium, and the maximum supporting prices are $p_1 = 1, p_2 = 2, p_3 = 2$, for a total revenue of $5$.  However, the revenue-maximizing \acron{} prices set the price of $\{a_1\}$ to $1$ and the price of $\{a_2,a_3\}$ to $R+2$, for a total revenue of $R+3$.  This gap can be made arbitrarily large by increasing $R$.

\section{Necessity of Bundling}
\label{app.bundling.necessary}

We show by way of example that the use of bundling is necessary for the statement of Theorem \ref{thm.superadd}, even if we relax the market-clearing requirement of Walrasian equilibrium and even for single-minded bidders.  We present an auction instance such that, for any set of item prices for which agent demand sets are disjoint, the social welfare of the resulting allocation is an $O(\sqrt{m})$ fraction of the optimal social welfare.

\begin{example}
Consider a combinatorial auction with $n = \sqrt{m}+1$ single-minded bidders, as follows.  One bidder desires set $M$ for a value of $m$.  Each of the remaining bidders $i$ desires a set $S_i$ of size $\sqrt{m}$, for a value of $1+2\sqrt{m}$.  These sets have the property that $|S_i \cap S_j| = 1$ for all $i \neq j$, and for every object $a \in M$, $a$ is contained in at most $2$ sets; the existence of such sets follows easily by induction on their size\footnote{That is, that there exist $t$ sets of size $t$ with this property, for all $t \leq \sqrt{m}$.}.

The optimal outcome allocates all objects to the first bidder, for a total value of $m$.  Suppose this optimal outcome can be supported by item prices $\mathbf{p}$.  Then $\sum_{a \in M} p_a \leq m$, but a counting argument demonstrates that for some set $S_i$, it must be that $\sum_{a \in S_i} p_a \leq 2\sqrt{m} < 1+2\sqrt{m}$.  Player $i$ would demand set $S_i$ at these prices, a contradiction.  Thus, at any pricing equilibrium, some set $S_i$ must be allocated to a buyer $i$; but this implies that no other set $S_j$ can be allocated, as $S_j \cap S_i \neq \emptyset$ for all $j \neq i$.  Thus the total social welfare at any equilibrium outcome with item prices is at most $1+2\sqrt{m}$.
\end{example}